\documentclass[letterpaper, 10 pt, conference]{ieeeconf}  % Comment this line out
\IEEEoverridecommandlockouts                              % This command is only
\usepackage{graphicx} % for pdf, bitmapped graphics files
\usepackage{amsmath} % assumes amsmath package installed
\usepackage{amssymb}  % assumes amsmath package installed

\usepackage{subfigure}
\usepackage{mathtools}
\usepackage{amsfonts}
\usepackage{verbatim}%用以多行注释\begin{comment}和\end{comment}
\newcommand{\etal}{\textit{et al}. }
\newcommand{\ie}{\textit{i}.\textit{e}., }
\newcommand{\eg}{\textit{e}.\textit{g}., }

\title{\LARGE \bf
Joint Optimization of the Deployment and Resource Allocation of UAVs in Vehicular Edge Computing and Networks}

\author{Yuke Zheng, Bo Yang and Cailian Chen% <-this % stops a space
}

\newtheorem{theorem}{\textbf{Theorem}}
\newtheorem{lemma}{\textbf{Lemma}}

\begin{document}

\maketitle
\thispagestyle{empty}
\pagestyle{empty}

%%%%%%%%%%%%%%%%%%%%%%%%%%%%%%%%%%%%%%%%%%%%%%%%%%%%%%%%%%%%%%%%%%%%%%%%%%%%%%%%
\begin{abstract}

With the development of smart vehicles, computing-intensive tasks are widely and rapidly generated. To alleviate the burden of on-board CPU, connected vehicles can offload tasks to or make request from nearby edge server thanks to the emerging Mobile Edge Computing (MEC). However, such approach may sharply increase the workload of an edge server, and cause network congestion, especially in rural and mountain areas where there are few edge servers. To this end, a UAV-assisted MEC system is proposed in this paper, and joint optimization algorithm of the deployment and resource allocation of UAVs (JOAoDR) is proposed to decide the location and balance the resource and rewards of the UAVs. We solve a long-term profit maximization problem in terms of the operator.
Numerical results demonstrated that our algorithm outperforms other benchmarks algorithm, and validated our solution.

\end{abstract}

%%%%%%%%%%%%%%%%%%%%%%%%%%%%%%%%%%%%%%%%%%%%%%%%%%%%%%%%%%%%%%%%%%%%%%%%%%%%%%%%
\section{INTRODUCTION}

In recent years, with the rapid development of Internet of Things (IoT), IoT have been rapidly shifting to application of artificial intelligence (AI) to transform smart devices, which generated massive computation data and transmission data.
Internet of Vehicles (IoV), as a special portion of IoT, have become smarter in supporting intelligent applications, such as on-board
cameras and embedded sensors, autonomous driving, intelligent platoon control, video-aided real-time navigation, interactive gaming, on-board Virtual Reality (VR) and Augmented Reality (AR) \cite{7994678,6616115,7908951,8543658}.

Different from common smart devices, the high mobility of vehicles cannot be overlooked, which indicated that data processing of vehicles need to be low delay and high reliability.
The on-board CPU was gradually overloading and cannot provide high-quality services \cite{7883826,8555636}.
\begin{comment}
Under this circumstance, Mobile Edge Computing (MEC) \cite{8650169,7901477,8315353} was proposed to tackle high transmission delay and network congestion.
This new paradigm brings computation and data storage closer to the location where it is needed, to improve response times and spectrum resources.
It allows the availability of the cloud servers inside or adjacent to the base station. The end-to-end latency perceived by the mobile terminal is therefore reduced with the MEC platform \cite{7883826,7833463}.
The purpose of MEC is in line with the core issues of the IoV, so it is regarded as a promising solution, and can be used as a basic architecture in IoV.

After the application of MEC in the IoV, due to the location characteristics of the MEC, the data in IoV can be stored nearby at or obtain nearby from an edge server (ES) closed to vehicles, so it can reduce the transmission delay and the probability of network congestion.
To be more specific, computation-intensive tasks generated by vehicles can be offloaded to nearby ES to process instead of occupying on-board CPU resources, which only have limited computation capability in general \cite{dai2018joint}.
Furthermore, offloading computation-intensive tasks to edge server is not the only advantage which MEC brings. It is also easier and faster for vehicles to get the required information from the cloud center, such as area high definition map \cite{8376252}, nearby traffic density data and personal traffic demand \cite{7974775}.
\end{comment}

Under this circumstance, Mobile Edge Computing (MEC) \cite{8650169,7901477,8315353} was proposed to tackle high transmission delay and network congestion.
This new paradigm brings computation and data storage closer to the location where it is needed, to improve response times and spectrum resources.
It allows the availability of the cloud servers inside or adjacent to the base station. The end-to-end latency perceived by the mobile terminal is therefore reduced with the MEC platform \cite{7883826,7833463}.
After the application of MEC in the IoV, computation-intensive tasks generated by vehicles can be offloaded to nearby ES to process instead of occupying on-board CPU resources, which only have limited computation capability in general \cite{dai2018joint}.
Furthermore, offloading computation-intensive tasks to edge server is not the only advantage which MEC brings. It is also easier and faster for vehicles to get the required information from the cloud center, such as area high definition map \cite{8376252}, nearby traffic density data and personal traffic demand \cite{7974775}.

However, once vehicles offload their computation-intensive tasks to a single edge server (which is commonly a base station) to relieve its own computing workload at the same time, the workload of the edge serve will rise sharply, especially in areas where the density of vehicles is relative high or during peak periods, causing latency and network congestion.
\begin{comment}
To cope with the issue and relieve workload of the base station, a promising paradigm that has received much attention lately called Vehicular Edge Computing (VEC). It enables infrastructure, such as Road Side Units (RSUs), as ESs to provide offloading service for vehicles to extend the computation capability to vehicular network \cite{dai2018joint,7907225}.
With the help of RSUs, Wang \etal \cite{8318667} put forward a feasible solution that enables offloading for real-time traffic management in fog-based IoV systems to minimize the average response time for events reported by vehicles.
Nasrin Taherkhani and Samuel Pierre \cite{7458837} proposed a centralized and localized data congestion control strategy to control data congestion using RSUs at intersections.

Such infrastructure must be deployed in fixed location. In some areas, such as rural areas and mountain regions, it is not easy to set up RSUs and maintain them. And because the vehicle distribution is different in different time period, such fixed-location deployment will lead to low service efficiency.
\end{comment}

Unmanned aerial vehicles (UAVs) have been witnessed as a promising approach for offering extensive coverage and additional computation capability to smart mobile devices. Compared with infrastructure-based VEC, UAV-assisted VEC possesses more reliable line-of-sight (LoS) links \cite{8594571,8422277}. the mobility of the UAV makes it easier to deploy in most areas, to improve quality of service (QoS) and to maintain, having advantage in saving cost \cite{8516294}.

\begin{comment}
In addition, thanks to the solar energy system and the rapid development of Wireless Power Transfer (WPT) technology \cite{8422277,7146165}, the application of energy-harvesting UAVs has taken a big step forward.
The solar energy system is a Off-grid power system included solar panels and DC-DC converter \cite{5175611}, which allow us to harness the power of the sun for electrical energy.
While as for WPT, there are three main technologies in general: inductive coupling-based WPT, magnetic resonant coupling-based WPT, and electromagnetic radiation-based WPT, in which the electromagnetic radiation-based WPT has longest transmission distance by using microwave frequencies.
Due to the location characteristics of the UAV, it is easy to meet the LoS link constraint of electromagnetic radiation between an UAV and the BS by setting up antennas on it which adopt silicon-controlled rectifier diode to collect microwave energy and convert it into direct current \cite{7146165}.
\end{comment}

In related works, Zhang \etal \cite{8594571} investigated a UAV-assisted mobile edge computing system with stochastic computation tasks.
The system aims to minimize the average weighted energy consumption
of smart mobile devices and the UAV, subject to the constraints on computation offloading, resource allocation, and flying trajectory
scheduling of the UAV.
Zhou \etal \cite{8422277} studied a UAV-enabled wireless powered MEC system and formulate a power minimization problem to minimize the energy consumption of the UAV.
Both of these works have destined initial and final locations for only one UAV.

While in this paper, we considered a UAV-assisted VEC system with only one BS and multiple UAVs.
Each UAV provide service in a given area.
Note that the coverage of different UAVs may partially overlap, and for one vehicle (user), it maybe within the coverage of several UAVs at the same time. Thus we investigated the task scheduling and cooperation among UAVs, as well as the deployment of them in each time period, aiming to maximize the long-term profit of the UAVs while balancing the energy consumption from the operator's perspective.

It is common that massive computation-intensive tasks and service requests are generated in a stochastic model for the operator, so that existing offloading strategies for deterministic tasks cannot be well applied.
Besides, in this paper, we consider new energy powered UAVs. The UAV can collect energy and recharge itself. The process of energy harvest also has stochastic nature, which cannot be ignored when a long-term performance is desired.%这里新加了几句，因为前面删掉了WPT
To this end, we utilize Lyapunov optimization \cite{6813406} to handle the issue where energy-efficient and profit-maximizing decisions must be made without knowing the future energy harvest or tasks arrival.

The main contributions of this paper are stated as follows.
\begin{itemize}
    \item We considered a UAV-assisted VEC system with only one BS and multiple UAVs where ground vehicles can generate tasks and service request.
    We characterized the uplink and downlink communication time delay simultaneously.
    Different from most works which ignore downlink transmission, as mentioned before, vehicles may request area high definition map, nearby traffic density data and so on, so we cannot ignore the effect of downlink.
    \item Unlike works towards stationary RSUs and other smart devices, the mobility of UAVs and ground vehicles is well considerd in this paper, which may have great influence on matching strategy and the deployment of UAVs.
    \item Due to the stochastic nature of both ground users' tasks arrival and UAVs' energy harvest, in order to obtain long-term profit maximization, by leveraging a Lyapunov-based approach, we balanced the profit and the remaining battery power of the UAVs.
    To sum up, we proposed the joint optimization algorithm of the deployment and resource allocation of UAVs (JOAoDR) to solve the formulated problem.
\end{itemize}

The rest of this paper is organized as follows.
Section \ref{sec2} describes the system model.
The long-term profit maximization problem is formulated in Section \ref{sec3}.
Section \ref{sec4} presents the Lyapunov-based approach to transform the origin problem and our solution.
Section \ref{sec5} shows the the numerical results.
Finally, Section \ref{sec6} concludes the paper.

\section{SYSTEM MODEL} \label{sec2}

\subsection{Network Architecture}

We considered a rural area with a base station BS whose signal covers the entire area.
There are $\left|\mathcal{B}\right|$ UAVs equipped with MEC servers having idle computing resources to help the BS to provide service within the area, denoted by set
$\mathcal{B}=\left\{ 1,2,\cdots,\left|\mathcal{B}\right| \right\}$
, where $|\cdot|$ denotes the cardinality of a set.
The coverage of each UAV is given and fixed.
Note that the 'coverage' here is not the exact wireless signal coverage but a fixed service range given in advance. 
The diameter of the 'coverage' is set to the radius of the exact wireless signal coverage so that each UAV can hover at any area within the given coverage, providing stable service.
The system is divided by time slots
$t\in\mathcal{T}=\left\{1,2,\cdots,T\right\}$, where $\tau$ denotes the length of a time slot.
Let $\mathcal{U}(t)=\left\{1,2,\cdots,\left|\mathcal{U}(t)\right|\right\}$ denote the set of ground vehicles within the coverage of the BS in time slot $t$.
It is assumed that the position of the vehicles do not change in one time slot, and may vary in different time periods. The historical distribution data of vehicles in the area is easily obtain, and the distribution probability density function is represented by $f_{x,y}(t)$.

As shown in Fig. \ref{fig1}, the coverage of UAVs can partially overlap each other. Therefore, vehicle $m$ may be in the service range of multiple UAVs at the same time. 
In Fig. \ref{fig1}, vehicle 4 can access to either UAV $2$ or UAV $3$ for task process at current time.
At the beginning of each time slot, computation tasks and traffic request may generated by vehicles in a stochastic manner. Then those vehicles transmit corresponding task information to the BS, which collects all of the information and acts as the dispatch center to determine the resource allocation of UAVs.

\textit{Notations:} In this paper, unless otherwise specified, we use $i$ and $m$ to identify the index of UAVs, vehicles with tasks, respectively. 

\begin{figure}[tpb]
    \centering
    \includegraphics[scale=0.18]{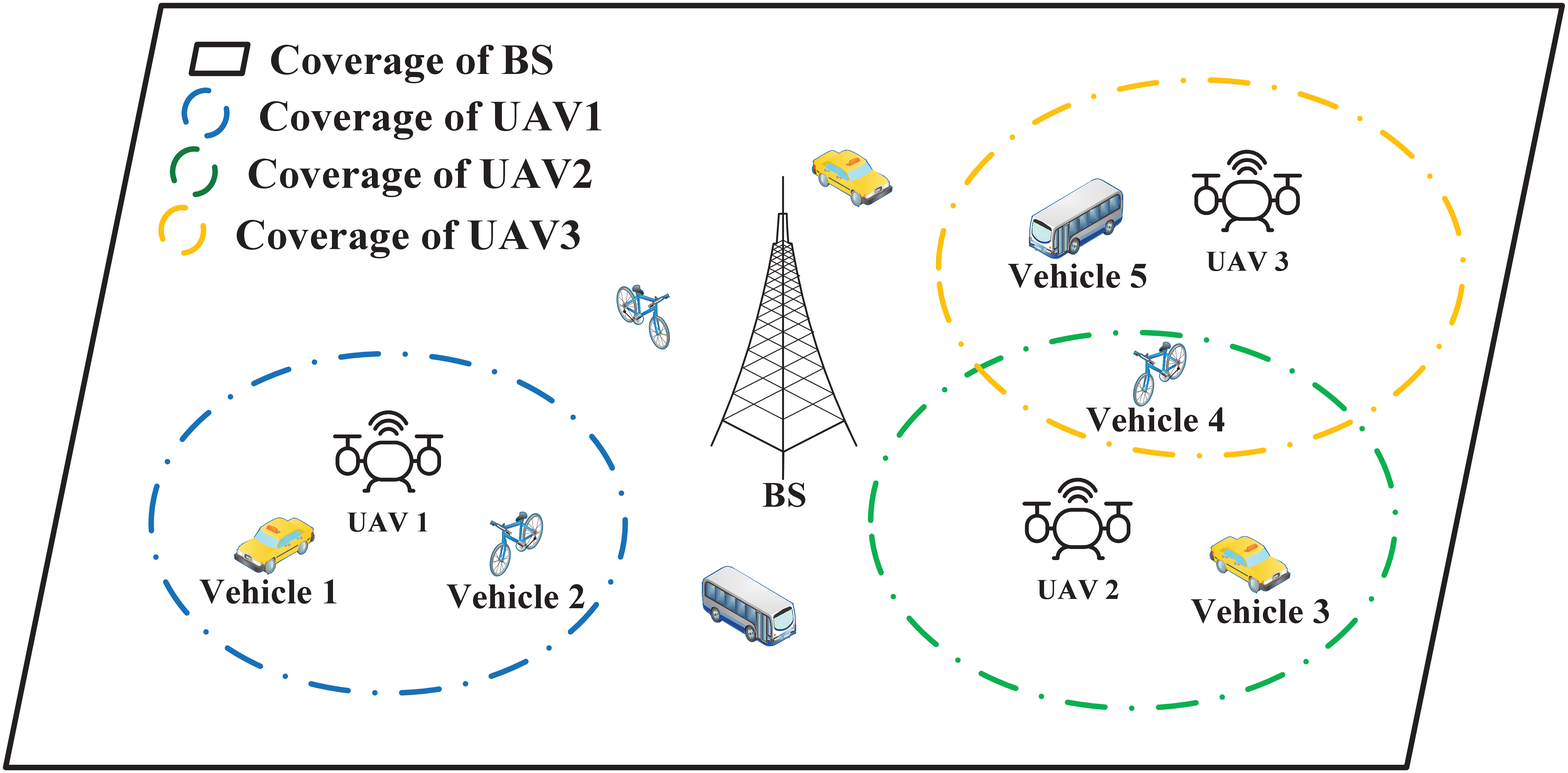}
    \caption{An example of network model}
    \label{fig1}
\end{figure}

\subsection{UAVs Trajectories Model}

Like most works, we assume that all the UAVs hovers at a fixed altitude $h$, which is determined by their mechanical parameters. And $\left(x_i(t),y_i(t)\right)$ can be used to represent the horizontal coordinate of UAV $i$ at time slot $t$. The maximum speed of UAVs denotes as $V_{max}$. Then in one time slot, the longest distance a UAV can fly is no more than $\tau V_{max}$, \ie
\begin{align}\label{eq4}
    \sqrt{(x_i (t+1)-x_i (t))^2+(y_i (t+1)-y_i (t))^2 }&\leq \tau V_{max} ,    \nonumber\\\forall  i&\in \mathcal{B}.
\end{align}

\subsection{Task Processing Model}

For simplicity of exposition, the path of vehicles are assumed to be straight and the speed are assumed to be fixed.

\subsubsection{Characteristics of the Computation Tasks}

At the beginning of each time period, each ground vehicle may generate a computing-intensive task or have traffic service request, and transmit the task information to the BS. The information can be expressed as a tuple $A_m (t) = \left \langle I_m (t), O_m (t), \phi_m (t), p_m (t),\lambda_m,(x_m(t),y_m(t)),\vec{v}_m(t)\right\rangle$, where $m$ represent the index of tasks in time slot t, $m \in \mathcal{M}(t)$. To be more specific, $I_m (t)$ and $ O_m (t)$ denote the input data size and output data size, respectively. $\phi_m(t)$ is the number of CPU cycles required to process the task. $p_m (t)$ denotes the remuneration to complete the task. $\lambda_m$ denotes the minimum acceptable quality of service (QoS), \ie transmission rate of downlink. While $(x_m(t),y_m(t))$ and $\vec{v}_m(t)$ are the position and velocity of vehicle $m$.

\subsubsection{Uplink Model}
Because the studied  scenario is a rural area and there are no high-rise occlusions, we modeled the air-to-ground propagation channel as a LoS link.
The freespace path loss model is adopted and the channel power gain
from vehicle $m$ to UAV $i$ can be expressed as
\begin{equation}
    h_{im}(t)=g_0 d_{im}^{-2}(t),
\end{equation}
where $g_0$ is the channel power gain at the reference distance $d = 1$ m. And $d_{im}(t)$ denotes the Euclidean distance between UAV $i$ and ground vehicle $m$ in time slot $t$, which are given by
$d_{im}(t)=\sqrt{(x_i(t)-x_m(t))^2+(y_i(t)-y_m(t))^2+h^2 }$.
Let $N_0$ denote the noise power. Furthermore, the Signal-to-Noise-Ratio (SNR) received at UAV $i$ from ground vehicle $m$ is
\begin{equation}
    \Gamma_{im}(t)=\frac{P_m(t) h_{im}(t)}{N_0},
\end{equation}
where $P_m(t)$ denotes the transmission power of the $m$th vehicle.
In this paper, for simplicity, we assume that there are enough wireless channels for communication, and the bandwidth of each channel is $W$ Hz. By applying Orthogonal Frequency Division Multiple Access (OFDMA) technique, each communication between UAVs and ground vehicles can access one channel.
Although such assumption may lead to waste of spectrum resources, in our scenario, \ie rural areas, there are enough idle resources to utilize.

According to Shannon–Hartley theorem, the maximum transmission rate from ground vehicle $m$ to UAV $i$ can be expressed as 
\begin{equation}
    R_{im}(t)=W \log_2(1+\Gamma_{im}(t)).
\end{equation}
Limited by coding methods and other reasons, the actual transmission rate $\bar{R}$ cannot reach the ideal rate, \ie
\begin{equation}
    \bar{R}_{im}(t)=\gamma R_{im}(t), \gamma \in (0,1).
\end{equation}
Once we obtain the actual transmission rate, the time slot delay caused by the uplink communication can be calculated by
\begin{equation}
    t_{im}^d(t)=\left \lceil \frac{I_m(t)}{\gamma \bar{R}_{im}(t)} \right \rceil,i\in \mathcal{B},m\in \mathcal{M}(t),
\end{equation}
where $\lceil\cdot\rceil$ denotes rounding up.

\subsubsection{Binary Decision Variables}

Let $\textbf{J}(t)=\{J_{im}(t)\},i\in \mathcal{B},m\in \mathcal{M}(t)$ denote the binary decision matrix, where $J_{im}(t)=1$ means that UAV $i$ is chosen to process task or request from ground vehicle $m$ at time slot $t$, otherwise $J_{im}(t)=0$.
Typically, a user's task can only be assigned to at most one UAV, \ie
\begin{equation}\label{eq1}
    \sum_{i\in\mathcal{B}}J_{im}(t)\leq 1.
\end{equation}
Furthermore, due to the strong mobility of ground vehicles, we must ensure that vehicle $m$ cannot move out the coverage of UAV $i$ before the input data is completely transmitted to the UAV $i$, otherwise $J_{im}(t)=0$.
To be more specific, define a set $\mathcal{B}_m(t)$, $i\in\mathcal{B}_m(t)$ represents vehicle $m$ is within the coverage of UAV $i$ in time slot $t$. Then, equation (\ref{eq1}) can be expanded to
\begin{align}\label{eq2}
    \sum_{i\in\mathcal{B}}J_{im}(t)= \sum_{i\in\mathcal{B}_m(t)}J_{im}(t)=\sum_{i\in\mathcal{B}_m(t+t_{im}^d(t))}J_{im}(t)\leq 1\\
    \sum_{i\notin\mathcal{B}_m(t)}J_{im}(t)=\sum_{i\notin\mathcal{B}_m(t+t_{im}^d(t))}J_{im}(t)=0,J_{im}(t)\in\{0,1\}
    \label{eq3}
\end{align}
From (\ref{eq2}) and (\ref{eq3}) we can obtain that whether UVA $i$ can handle user $m$'s task at the current time period is not only affected by the vehicle position at the current time, but also by the vehicle position after $t_{im}^d(t)$ time slots.
%求解里把如何确定Bm(t)写上

\subsubsection{QoS}

Let $\textbf{P}(t)=\{P_{im}(t)\},i\in \mathcal{B},m\in \mathcal{M}(t)$ denote downlink transmission power matrix, where $P_{im}(t)$ is a optimization variable, representing the transmission power from UAV $i$ to ground vehicle $m$.
Once UAV $i$ is chosen to process task or request from ground vehicle $m$ at time slot $t$, \ie $J_{im}(t)=1$, it has to wait for $t_{im}^d(t)$ time slots until the task is completely uploaded. Then UAV $i$ computes the task and returns output data to vehicle $m$ within only one time slot, otherwise vehicle $m$ may move out of the coverage, which may cause connection loss.
Likewise, the downlink transmission rate
\begin{equation}
    r_{im}(t)=\gamma W \log_2\left(1+\frac{P_{im}(t)g_0}{d_{im}^2 \!\! (t+t_{im}^d(t))N_0} \right),
\end{equation}
where $d_{im} \!\! \left(t+t_{im}^d(t)\right)$ is the distance between UAV $i$ and vehicle $m$ in time slot $t+t_{im}^d(t)$.
The transmission rate must satisfy the the minimum acceptable QoS of user $m$, \ie
\begin{equation}\label{eq5}
    r_{im}(t)\geq \lambda_mJ_{im}(t).
\end{equation}
Let $s_{im}(t)$ denote the CPU speed (in cycle per second) of UAV $i$ to compute task from vehicle $m$ generated in time slot $t$. Note that if $J_{im}(t)=1$, $s_{im}(t)$ is in fact the CPU speed for corresponding task after $t_{im}^d(t)$ slots, because of uploading delay.
Due to the powerful computation capacity and transmission capacity compared to ground vehicles, it’s assumed that a UAV can adjust computing frequency and transmission power to finish the task within one time slot $t$, \ie
\begin{equation}\label{eq6}
    \left(\frac{\phi_m(t)}{s_{im}(t)}+\frac{O_m(t)}{r_{im}(t)}\right)\cdot J_{im}(t)\leq \tau
\end{equation}

\subsubsection{UAV Energy Consumption}

UAV energy consumption can be divided into three parts, namely receiving energy consumption, calculating energy consumption and transmitting energy consumption. 
Suppose UAV $i$ serves vehicle $m$. 
The receiving energy consumption is given by $E_{im}^{rec}(t)=C_i^rI_m(t)$, where $C_i^r$ denotes the energy consumption of UAV $i$ for receiving one-unit input data size from ground vehicle $m$.
The calculating power consumption is given by
\begin{equation}
    p_{im}^{cpu}(t)=\alpha_is_{im}^3(t)+\beta_i
\end{equation}
by adopting a simple computing model \cite{6848173,8651320}, where $\alpha_i$ and $\beta_i$ are two parameters determined by UAV $i$'s CPU system.
We can further derive the calculating energy consumption by
\begin{equation}
    E_{im}^{cpu}(t)=p_{im}^{cpu}(t)\cdot\frac{\phi_m(t)}{s_{im}(t)}.
\end{equation}
Finally, the transmitting energy consumption can be obtained by 
\begin{equation}
    E_{im}^{snd}(t)=P_{im}(t)\cdot\frac{O_m(t)}{r_{im}(t)}.
\end{equation}
To sum up, the energy consumption of UAV $i$ for serving vehicle $m$ is
\begin{equation}
    E_{im}(t)=E_{im}^{rec}(t)+E_{im}^{cpu}(t)+E_{im}^{snd}(t).
\end{equation}

\subsubsection{Processing Capability}

The parallel computing capability of a UAV is limited, and the number of channels allocated to each UAV is limited, too. Without loss of generality, we supposed that the maximum number of tasks that UAV $i$ can perform parallel at the same time is equal to the number of channels allocated to it, denoted as $c_i$.
Once UAV $i$ is chosen to serve groung vehicle $m$ in time slot $t$, the corresponding channel will be occupied until $t+t_{im}^d(t)+1$, because $t_{im}^d(t)$ slots are needed to upload input data and one slot is for UAV to compute task and return output data.
Let $y_i(t)$ denote the number of channels which are occupied during time slot $t$.
Accordingly, 
\begin{comment}
\begin{equation*}\label{eq15}
    y_i(t)=\sum_{m\in\mathcal{M}(t)}J_{im}(t)+\sum_{m\in\mathcal{M}(t-1)}F (t_{im}^d(t-1))\cdot J_{im}(t-1)+\cdots
    + \sum_{m\in\mathcal{M}(0)}F (t_{im}^d(0)-(t-1))\cdot J_{im}(0)
    \mathrlap{=\sum_{k=0}^{t-1}\sum_{m\in\mathcal{M}(k)}F (t_{im}^d(k)-(t-k-1))\cdot J_{im}(k)+\sum_{m\in\mathcal{M}(t)}J_{im}(t)}
\end{equation*}
\end{comment}
\begin{comment}
\begin{alignat}{2}\label{eq15}
    y_i(t)&=\sum_{m\in\mathcal{M}(t)}J_{im}(t)&&+\!\!\!\!\sum_{m\in\mathcal{M}(t-1)}\!\!\!\!F (t_{im}^d(t-1))\cdot J_{im}(t-1)+\cdots\nonumber\\
    &&&+ \sum_{m\in\mathcal{M}(0)}F (t_{im}^d(0)-(t-1))\cdot J_{im}(0)\nonumber\\
    &\mathrlap{=\sum_{k=0}^{t-1}\sum_{m\in\mathcal{M}(k)}F (t_{im}^d(k)-(t-k-1))\cdot J_{im}(k)+\sum_{m\in\mathcal{M}(t)}J_{im}(t)}
\end{alignat}
\end{comment}

\begin{alignat}{2}\label{eq15}
    y_i(t)&=\!\!\!\sum_{m\in\mathcal{M}(t)}\!\!\!J_{im}(t)&&+\!\!\!\!\sum_{m\in\mathcal{M}(t-1)}\!\!\!\!F (t_{im}^d(t-1))\cdot J_{im}(t-1)\nonumber\\
    &&&+\cdots\nonumber\\
    &&&+\!\!\! \sum_{m\in\mathcal{M}(0)}F (t_{im}^d(0)-(t-1))\cdot J_{im}(0)\nonumber\\
    &\mathrlap{=\sum_{k=0}^{t-1}\sum_{m\in\mathcal{M}(k)}F (t_{im}^d(k)-(t-k-1))\cdot J_{im}(k)}\nonumber\\
    &&&+\sum_{m\in\mathcal{M}(t)}J_{im}(t)
\end{alignat}

where we define $F(x)=\left\{\begin{aligned}
1,x>0\\
0,x\leq 0
\end{aligned}\right.$.
The number of channels which are occupied during time slot $t$ can not exceed the maximum number of channels, \ie
\begin{equation}\label{eq7}
    y_i(t)\leq c_i, \forall i\in\mathcal{B}.
\end{equation}

\subsection{Dynamic Battery Power Model}

UAVs can charge themselves and the charging process is stochastic process. Let $\eta_i(t)$ be the electrical energy collected by UAV $i$ through solar conversion
%and WPT technique
during time period $t$, which is upper bounded by $\eta_{max}$.
Let $\textbf{E}(t)=\{E_i(t)\},i\in\mathcal{B}$ be the dynamic battery energy queue vector. 

How much the UAV can charge at the current slot depends on its current battery power and the energy it collects.
To be more specific, the amount of energy which UAV $i$ charge itself is
\begin{equation}
    e_i(t)=\min (\theta_i-E_i(t),\eta_i(t)),
\end{equation}
where $\theta_i$ is the desired battery energy corresponding to UAV $i$. It means that UAV $i$ will charge itself in every slot until the desired battery energy is reached.
Then the dynamic of UAV $i$'s battery energy is
\begin{equation}\label{eq8}
    E_i(t+1)=E_i(t)+e_i(t)-\sum_{m\in\mathcal{M}(t)}J_{im}(t)E_{im}(t)
\end{equation}
Note that $E_i(t)$ is not exactly the battery energy of UAV $i$ in practice, because it predict the energy consumption $E_{im}(t)$ after $t_{im}^d(t)$ slots. But it can measure the remaining battery energy as much as possible.
Without loss of generality, we assume that the batteries of UAVs are full of charge initially, \ie $E_i(0)=\theta_i,\forall i\in\mathcal{B}$.
For simplicity, we assume UAVs charge themselves at the end of each slot.
Further we can obtain the energy consumption constraint in time slot $t$
\begin{equation}\label{eq9}
    E_i(t)\geq\sum_{m\in\mathcal{M}(t)}J_{im}(t)E_{im}(t),\quad\forall i\in\mathcal{B}.
\end{equation}

\section{PROBLEM FORMULATION} \label{sec3}

In this paper, we focus on maximizing the long-term profit of UAVs.
The remuneration of UAVs in time slot $t$ is given by
\begin{equation}
    R(t)=\sum_{i\in\mathcal{B}}\!\sum_{m\in\mathcal{M}(t)}\!J_{im}(t)p_m(t)
\end{equation}
Then the problem can be formulated as follows
\[\begin{gathered}\textbf{P1:}\max_{\textbf{P}(t),\textbf{s}(t),\textbf{J}(t), \textbf{L}(t)}\lim_{T\to\infty}\frac{1}{T+1}\sum_0^T\mathbb{E}\!\left\{\! \sum_{i\in\mathcal{B}}\!\sum_{m\in\mathcal{M}(t)}\!\!\!\!\!J_{im}(t)p_m(t)\!\right\}\\
\text{s.t. }(\ref{eq4}),(\ref{eq2}),(\ref{eq3}),(\ref{eq5}),(\ref{eq6}),(\ref{eq7}),(\ref{eq8}) \text{ and } (\ref{eq9})\\
s_{im}(t)\leq s_{i,max},i\in\mathcal{B},m\in\mathcal{M}(t)
\end{gathered}\]

In the above formulation,the optimization variable matrix $\textbf{s}(t)=\{s_{im}(t)\},i\in\mathcal{B}, m\in\mathcal{M}(t)$ denotes the UAVs' CPU speed matrix in time slot $t$, while $\textbf{L}(t)=\{\left(x_i(t),y_i(t)\right)\}, i\in\mathcal{B}$ repersents the horizontal coordinate matrix of UAVs in time slot $t$.
The last constraint shows that the CPU speed of UAV $i$ allocated for each task cannot exceed $s_{i,max}$ because of the limited parallel computing capacity.

Solving such an optimization problem requires not only the decision of the current time slot, but also the decision of the future, which is difficult to solve without knowing in advance the energy collection and task arrival in the future. In the next section, we apply a Lyapunov-based approach to transform the long-term optimization into single-slot optimizations that can be solved separately. Then, our JOAoDR is proposed to solve the problem.

\section{PROBLEM TRANSFORMATION AND SOLUTION} \label{sec4}

In this section, in order to solve the long-term maximization problem which is difficult to analyze, we first transform the original optimization problem \textbf{P1} into several single-slot optimizations base on Lyapunov optimization, by which we removed the relevance of the problem in continuous time.
However, the transformed problem is still hard to solve due to the strong coupling between variables.
Accordingly, we divided the transformed problem into two stages,\ie online solving stage and offline solving stage.

\subsection{Problem Transformation}

For simplicity to express, we define $\textbf{Q}(t)=\{Q_i(t)\},i\in\mathcal{B}$, where $Q_i(t)=\theta_i -E_i(t)$.
To ensure the stability of UAVs batteries power , we define the Lyapunov function as
\begin{equation}
    L(t)=\frac{1}{2}\sum_{i\in\mathcal{B}}Q_i^2(t)
\end{equation}
This definition intuitively means that we expect the battery power of UAV $i$ to be as close to the corresponding parameter $\theta_i$ as possible by minimizing the drift of the Lyapunov function.
The Lyapunov drift can be defined as
\begin{equation}\label{eq10}
    \Delta(t)=\mathbb{E}\left\{L(t+1)-L(t)|\textbf{Q}(t)\right\}.
\end{equation}

By adding the penalty function (subtracting the profit of UAVs ) on both sides of (\ref{eq10}), the drift-plus-penalty (drift-minus-reward) function can be given by
\begin{align}
    \Delta_V(t)&=\Delta(t)-V\mathbb{E}\{R(t)|\textbf{Q}(t)\}\nonumber\\
    &=\mathbb{E}\left\{L(t+1)-L(t)-VR(t)|\textbf{Q}(t)\right\},
\end{align}
where $V$ is a control parameter to deal with the tradeoff between UAVs profit and batteries power.
After incorporate the UAVs profit into the drift-plus-penalty function, we can transform the origin optimization problem into minimizing $\Delta_V(t)$ at each time slot.
\begin{theorem}
The given drift-plus-penalty function $\Delta_V(t)$ is upper bounded by
\begin{align}\label{eq13}
    &\Delta_V(t)\leq A-\mathbb{E}\left\{\sum_{i\in\mathcal{B}}Q_i(t)e_i(t)|\textbf{Q}(t)\right\}\nonumber\\
    \!\!&-\!\!\mathbb{E}\left\{\sum_{i\in\mathcal{B}}\!\sum_{m\in\mathcal{M}(t)}\!\!\!\!\left[VJ_{im}\!(t)p_m\!(t)-Q_i\!(t)J_{im}\!(t)E_{im}\!(t)\right]|\textbf{Q}(t)\right\},
\end{align}
where $A=\frac{|\mathcal{B}|}{2}\eta_{max}^2+\frac{1}{2}\sum_{i\in\mathcal{B}}c_iE_{max}$ is a constant.
\end{theorem}
\begin{proof}
Let's define $L_i(t) = \frac{1}{2}Q_i^2(t),\forall i\in\mathcal{B}$ and $\Delta_i(t)=\mathbb{E}\left\{L_i(t+1)-L_i(t)|\textbf{Q}(t)\right\}$. Substituting the first equation into the second yields
\begin{align}\label{eq11}
    \Delta_i(t)=&\frac{1}{2}\mathbb{E}\left\{\left[E_i^2(t+1)-E_i^2(t)\right]|\textbf{Q}(t)\right\}\nonumber\\
    &-\theta_i\cdot \mathbb{E}\left\{e_i(t)-\sum_{m\in\mathcal{M}(t)}J_{im}(t)E_{im}(t)|\textbf{Q}(t)\right\}.
\end{align}
Substituting (\ref{eq8}) into (\ref{eq11}), we omit the time indication $(t)$ without affecting expression for simplicity and get

\begin{comment}
\begin{align}\label{eq12}
    \Delta_i&=\frac{1}{2}e_i^2+\frac{1}{2}\left(\sum_{m\in\mathcal{M}}J_{im}E_{im}\right)^2- \mathbb{E}\left\{ Q_ie_i |\textbf{Q}(t)\right\} +\mathbb{E}\left\{(Q_i-e_i)\sum_{m\in\mathcal{M}}J_{im}E_{im}|\textbf{Q}(t)\right\}\nonumber\\
    &\leq A_i-\mathbb{E}\left\{Q_ie_i|\textbf{Q}(t)\right\}+\mathbb{E}\left\{Q_i\sum_{m\in\mathcal{M}}J_{im}E_{im}|\textbf{Q}(t)\right\},
\end{align}
\end{comment}

\begin{alignat}{2}\label{eq12}
    \Delta_i&=\frac{1}{2}e_i^2&&+\frac{1}{2}\left(\sum_{m\in\mathcal{M}}J_{im}E_{im}\right)^2- \mathbb{E}\left\{ Q_ie_i |\textbf{Q}(t)\right\}\nonumber\\ &&&+\mathbb{E}\left\{(Q_i-e_i)\sum_{m\in\mathcal{M}}J_{im}E_{im}|\textbf{Q}(t)\right\}\nonumber\\
    &\mathrlap{\leq A_i-\mathbb{E}\left\{Q_ie_i|\textbf{Q}(t)\right\}}\nonumber\\&&&+\mathbb{E}\left\{Q_i\sum_{m\in\mathcal{M}}J_{im}E_{im}|\textbf{Q}(t)\right\},
\end{alignat}

where $A_i=\frac{1}{2}e_i^2+\frac{1}{2}\left(\sum_{m\in\mathcal{M}}J_{im}E_{im}\right)^2$.

By summing (\ref{eq12}) from $i=1$ to $i=|\mathcal{B}|$ and adding $-V\mathbb{E}\{R(t)|\textbf{Q}(t)\}$ on both sides of (\ref{eq12}), we obtain
\begin{align}
    \Delta_V\leq&A-\mathbb{E}\left\{\sum_{i\in\mathcal{B}}Q_ie_i|\textbf{Q}\right\}\nonumber\\
    &+\mathbb{E}\left\{\sum_{i\in\mathcal{B}}\!\sum_{m\in\mathcal{M}}Q_iE_{im}J_{im}-VJ_{im}p_m|\textbf{Q}\right\}.
\end{align}
Rearranging the terms yields (\ref{eq13}).
\end{proof}

Then, minimizing the drift-plus-penalty function $\Delta_V(t)$ is equivalent to minimizing the right-hand-side (RHS) of (\ref{eq13}).
We can further solve the following optimization problem
\[\begin{gathered}\textbf{P2:}\max_{\textbf{P}(t),\textbf{s}(t),\textbf{J}(t), \textbf{L}(t)}\sum_{i\in\mathcal{B}}\!\sum_{m\in\mathcal{M}(t)}\!\!\!\!J_{im}\!(t)\left[Vp_m\!(t)-Q_i\!(t)E_{im}\!(t)\right]\\
\text{s.t. }(\ref{eq4}),(\ref{eq2}),(\ref{eq3}),(\ref{eq5}),(\ref{eq6}),(\ref{eq7}) \text{ and } (\ref{eq9})\\
s_{im}(t)\leq s_{i,max},i\in\mathcal{B},m\in\mathcal{M}(t)
\end{gathered}\]

\subsection{Online Optimization}
 
In the aforementioned problem \textbf{P2}, the term $v_{im}(t)\triangleq Vp_m\!(t)-Q_i\!(t)E_{im}\!(t)$ can be regarded as the weight of $J_{im}(t)$. 
The optimal deployment of UAVs $\textbf{L}(t)$ is difficult to get an explicit solution in such problem.
In the following, we first obtained the optimal solution of $\textbf{P}(t),\textbf{s}(t)$ and $\textbf{J}(t)$ with fixed $\textbf{L}(t)$.
Then we provided a feasible offline solution to optimize $\textbf{L}(t)$ in next subsection.

We extract a sub-problem \textbf{P3} from problem \textbf{P2}, \ie
\[\begin{gathered}\textbf{P3:}\min_{P_{im}(t),s_{im}(t)}E_{im}(t)\\
\text{s.t. }(\ref{eq5}) \text{ and }(\ref{eq6}) \\
s_{im}(t)\leq s_{i,max}
\end{gathered}\]
\begin{theorem}\label{th2}
The optimal solution for \textbf{P3} is necessary conditions for the optimal solution for \textbf{P2}. \ie for those $J_{im}^*(t)=1$ in \textbf{P2}, the corresponding $s_{im}^*(t)$ and $P_{im}^*(t)$ are equal to the optimal solution $s_{im}^{**}(t)$ and $P_{im}^{**}(t)$ for \textbf{P3}.
\end{theorem}
\begin{proof}
We prove by contradiction.
Assume that the optimal solution of \textbf{P2} are $\textbf{P}^*(t),\textbf{s}^*(t),\textbf{J}^*(t)$, and without loss of generality, we assume $J_{11}^*(t)=1$ and the corresponding $s_{11}^*(t),P_{11}^*(t)$ are not equal to $s_{11}^{**}(t),P_{11}^{**}(t)$, respectively, in \textbf{P3} with $i=1,m=1$.
So that there exists at least one feasible pair $\left\langle\bar{s}_{11}(t),\bar{P}_{11}(t)\right\rangle$ (\eg $\left\langle s_{11}^{**}(t),P_{11}^{**}(t)\right\rangle$) such that $\bar{E}_{11}(t)<E_{11}^*(t)$ in \textbf{P3}.

Now we set a solution pair for \textbf{P2}, with $\bar{\textbf{J}}(t)=\textbf{J}^*(t),\bar{P}_{im}(t)=P_{im}^*(t),\bar{s}_{im}(t)=s_{im}^*(t),i\in \mathcal{B},m\in\mathcal{M}(t) \text{ except } \bar{s}_{11}(t) \text{ and } \bar{P}_{11}(t)$.

It is obvious to verify that the new solution pair is feasible.
Besides, Mathematical Induction (MI) can be used to verify that  $\theta_i\geq E_i(t)$.
So that $Q_i(t)\geq 0$.
Because $\bar{E}_{11}(t)<E_{11}^*(t)$ and $Q_i(t)\geq 0$, $\bar{v}_{im}(t)>v_{im}^*(t)$. Thus $\sum_{i\in\mathcal{B}} \sum_{m\in\mathcal{M}(t)}\bar{J}_{im}(t)\bar{v}_{im}(t)> \sum_{i\in\mathcal{B}}\sum_{m\in\mathcal{M}(t)} J_{im}^*(t)v_{im}^*(t)$, which contradicts with $\textbf{P}^*(t),\textbf{s}^*(t),\textbf{J}^*(t)$ being the optimal solution of \textbf{P2}.
\end{proof}

Through \textit{theorem \ref{th2}}, we can first solve the sub-problem \textbf{P3} for all $i\in\mathcal{B}$ and $m\in\mathcal{M}(t)$ in each slot and then substitute the optimal values $\left\langle s_{im}^{**}(t),P_{im}^{**}(t)\right\rangle$ into \textbf{P2}.

However, the Hesse matrix of function $E_{im}(t)$ is not positive definite or semi-positive definite, so $E_{im}(t)$ is not a convex function, leading to \textbf{P3} a non-convex problem.
Traditional Karush-Kuhn-Tucker (KKT) conditions cannot be applied to find the optimal solution.
To solve \textbf{P3}, we first decouple two variables $s_{im}(t)$ and $P_{im}(t)$. \ie we neglect constraint (\ref{eq6}).
\begin{lemma}\label{le1}
Without constraint (\ref{eq6}), the optimal solution for \textbf{P3} is given by
\[\left\{\begin{aligned}
&s_{im}^{**}(t)=\min\left(\sqrt[3]{\frac{\beta_i}{2\alpha_i}},s_{i,max}\right)\\
&P_{im}^{**}(t)=\left(2^{\frac{\lambda_mJ_{im}\!(t)}{\gamma W}}-1\right) \frac{N_0 d_{im}^2\!(t+t_{im}^d(t))}{g_0}
\end{aligned}\right.
\]
\end{lemma}
\begin{proof}
Now that \textbf{P3} can be divided into two sub-problems regarding $s_{im}(t)$ and $P_{im}(t)$, respectively.
\begin{itemize}
    \item[1)] As for $s_{im}(t)$, the sub-problem can be formulated as
    \begin{gather*}
        \textbf{P4:}\min_{s_{im}(t)}E_{im}^{cpu}(t)= \left(\alpha_is_{im}^3(t)+\beta_i\right)\cdot\frac{\phi_m(t)}{s_{im}(t)}\\
        \text{s.t. }s_{im}\leq s_{i,max}
    \end{gather*}
    \textbf{P4} is a convex optimization and the optimal solution $s_{im}^{**}(t)$ can be obtained by KKT conditions.
    \item[2)] As for $P_{im}(t)$, the sub-problem can be formulated as
    \begin{gather*}
        \textbf{P5:}\min_{P_{im}(t)}E_{im}^{snd}(t)=\frac{P_{im}(t)O_m(t)}{\gamma W\log_2\left(1+\frac{P_{im}(t)g_0}{d_{im}^2 \!\! (t+t_{im}^d(t))N_0} \right)} \\
        \text{s.t. }
        \gamma W \log_2\left(1+\frac{P_{im}(t)g_0}{d_{im}^2 \!\! (t+t_{im}^d(t))N_0} \right)\geq \lambda_mJ_{im}(t)
    \end{gather*}
    When $J_{im}(t)=0$, the constraint of \textbf{P5} can be naturally satisfied. 
    At this time UAV $i$ is not chosen to serve vehicle $m$. So we can easily set $P_{im}^{**}(t)=0$.
    Thus here the situation only when $J_{im}(t)=1$ is considered.
    
    Taking derivative of $E_{im}^{snd}(t)$, we can get 
    \[
    \frac{dE_{im}^{snd}(t)}{dP_{im}(t)}= \frac{ 
     \log_2\left(\ 1+\frac{P_{im}(t)}{\mu_{im}(t)}\right) - \frac{P_{im}(t)} {\left(P_{im}(t) + \mu_{im}(t)\right]\ln2} } {\frac{1}{\gamma W O_m(t)}
    \left[ \gamma W \log_2\left(\ 1+\frac{P_{im}(t)}{\mu_{im}(t)}\right)\right]^2},
    \]
    where $\mu_{im}(t)\triangleq \frac{N_0d_{im}^2 \! (t+t_{im}^d(t))}{g_0}$.
    Define $F_{im}(t)=\log_2\left(\ 1+\frac{P_{im}(t)}{\mu_{im}(t)}\right) - \frac{P_{im}(t)} {\left(P_{im}(t) + \mu_{im}(t)\right)\ln2}$, then
    \[
    \frac{dF_{im}(t)}{dP_{im}(t)}=\frac{P_{im}(t)}{\left( P_{im}(t)+\mu_{im}(t)\right)^2\ln2}.
    \]
    We can get $\frac{dF_{im}(t)}{dP_{im}(t)}>0$ when $P_{im}(t)>0$.
    So that $F_{im}(t)$ is monotonically increase with $P_{im}(t)$.
    $\left.F_{im}(t)\right|_{P_{im}(t)=0}=0$, thus $F_{im}(t)>0$, \ie $\frac{dE_{im}^{snd}(t)}{dP_{im}(t)}>0$.
    It is proved that $E_{im}^{snd}(t)$ is monotonically increase with $P_{im}(t)$.
    By rearranging the constraint term, we can get $P_{im}(t)\geq\left(2^{\frac{\lambda_mJ_{im}\!(t)}{\gamma W}}-1\right) \frac{N_0 d_{im}^2\!(t+t_{im}^d(t))}{g_0}$, through which we can get the optimal solution.
\end{itemize}
\end{proof}

\textit{Lemma \ref{le1}} gives the optimal solution of $s_{im}(t)$ and $P_{im}(t)$ when constraint (\ref{eq6}) is satisfied. However, sometimes the optimal solution in \textit{Lemma \ref{le1}} is not feasible in \textbf{P3}.
To this end, we provide the following theorem.
\begin{theorem}
When the solution cannot meet the constraint (\ref{eq6}), the optimal solution must be at the constraint boundary, \ie \begin{equation}\label{eq14}
\left(\frac{\phi_m(t)}{s_{im}^{**}(t)}+\frac{O_m(t)}{r_{im}^{**}(t)}\right)= \tau
\end{equation}
 %当约束不能满足，最优值一定在约束边界
\end{theorem}
\begin{proof}
We prove by contradiction. Assume $\left(\frac{\phi_m(t)}{s_{im}^{**}(t)}+\frac{O_m(t)}{r_{im}^{**}(t)}\right)< \tau$, then we can find a proper value of $\bar{P}_{im}(t)$ such that $\left(\frac{\phi_m(t)}{s_{im}^{**}(t)}+\frac{O_m(t)}{\bar{r}_{im}(t)}\right)= \tau$, where $\bar{r}_{im}(t)=\gamma W \log_2\left(1+\frac{\bar{P}_{im}(t)g_0}{d_{im}^2 \!\! (t+t_{im}^d(t))N_0} \right)$.
Note that $\bar{P}_{im}(t)<P_{im}^{**}(t)$ because $r_{im}(t)$ is monotonically increase with $P_{im}(t)$.
According to the aforementioned lemma, $E_{im}^{snd}(t)$ is monotonically increase with $P_{im}(t)$.
So it can be obtained that $\bar{E}_{im}^{snd}(t)<E_{im}^{snd\,**}(t)$, further, $\bar{E}_{im}(t)<E_{im}^{**}(t)$, which contradicts with $s_{im}^{**}(t),P_{im}^{**}(t)$ being the optimal solution of \textbf{P3}.
\end{proof}

To sum up, we solve \textbf{P3} by the following steps. First we adopt the values of $s_{im}(t)$ and $P_{im}(t)$ in \textit{lemma \ref{le1}} and check the feasibility.
If constraint (\ref{eq6}) is satisfied, the adopted values are the optimal solution.
Otherwise, we substitute (\ref{eq14}) into \textbf{P3} and formulate a one-dimensional optimization problem, which can be solved via well-known methods. 

With the optimal solution $E_{im}^{**}(t)$ given, we can further transform \textbf{P2}, \ie
\[
\begin{gathered}\textbf{P6:}\max_{\textbf{J}(t),}\sum_{i\in\mathcal{B}}\!\sum_{m\in\mathcal{M}(t)}\!\!\!\!J_{im}\!(t)v_{im}^{**}(t)\\
\text{s.t. }(\ref{eq2}),(\ref{eq3}),(\ref{eq7}) \text{ and } (\ref{eq9})
\end{gathered}
\]
where $v_{im}^{**}(t)=Vp_m(t)-Q_i(t)E_{im}^{**}(t)$ is given, and it can be regarded as the weight of $J_{im}(t)$. 
While with these constraint, \textbf{P6} is a non-standard assignment problem.
We aim to convet it to a standard assignment problem such that classic methods can be applied.
\begin{theorem}\label{th4}
By setting $\theta_i=\frac{Vp_{max}}{E_{i,min}}+c_iE_{i,max},i\in\mathcal{B}$, where $p_{max}$, $E_{i,min}$ and $E_{i,max}$ denote the maximum payment from users, minimum and maximum energy consumption of UAV $i$ to serve one vehicle, respectively, the constraint (\ref{eq9}) is indeed redundant.
\end{theorem}
\begin{proof}
By setting $\theta_i$, if $E_i(t)<c_iE_{i,max}$, then
\begin{align*}
    v_{im}(t)&=Vp_m(t)-\left( \theta_i-E_i(t)\right)E_{im}(t)\\
    &=Vp_m(t)-\left( \frac{Vp_{max}}{E_{i,min}}+c_iE_{i,max}-E_i(t)\right)E_{im}(t)\\
    &<Vp_{max}-\frac{Vp_{max}}{E_{i,min}}\cdot E_{im}(t)\leq0
\end{align*}
That is to say, the optimal $J_{im}^*(t)=0,m\in\mathcal{M}(t)$.
Then (\ref{eq9}) is satisfied.
On the other hand, note that $\sum_{m\in\mathcal{M}(t)}J_{im}(t)E_{im}(t)\leq c_iE_{im}(t)\leq c_iE_{i,max}$ because of constraint (\ref{eq7}).
Thus if $E_i(t)\geq c_iE_{i,max}$, constraint (\ref{eq9}) is satisfied, too.
To sum up, the constraint (\ref{eq9}) is indeed redundant.
\end{proof}

We still need some procedures to convert the problem to standard assignment problem.
\begin{itemize}
    \item [1)]Based on (\ref{eq3}), whether vehicle $m$ is within the coverage of UAV $i$ or not must be judged.
    The location $l_m(t+t_{im}^d(t))\triangleq\left(x_m(t+t_{im}^d(t)),y_m(t+t_{im}^d(t))\right)$ of vehicle $m$ is given by 
    \begin{equation}
        l_m(t+t_{im}^d(t))=l_m(t)+\vec{v}_m(t)\tau t_{im}^d(t).
    \end{equation}
    At each slot, the algorithm checks whether $l_m(t),l_m(t+t_{im}^d(t))\in C_i$, where $C_i$ denotes the coverage of UAV $i$.
    If either one not so, set $v_{im}(t)=0$, thus we make sure the optimal $J_{im}^*(t)=0$, which satisfies (\ref{eq3}).
    \item [2)]At each slot, Eq (\ref{eq15}) is applied to update $y_i(t)$ and $\sum_{m\in\mathcal{M}(t)}J_{im}(t)$ is determined by (\ref{eq7}).
\end{itemize}

The transformed problem is a standard assignment problem which the Hungarian algorithm \cite{Kuhn2010} can be utilized to solve.

\subsection{Offline Optimization} 

We focus on finding an indicator that measure the efficiency of UAVs' position. The position of the UAVs cannot be determined through online procedures, because we must grasp the task arrival situation and vehicle position distribution at the current moment if online optimization is applied, which is unrealistic in practice. In the actual situation, a UAV cannot know the stochastic task arrival situation until it reach a place and provide service. Therefore, the deployment of the UAVs must be determined through offline calculations in advance. Here, the historical distribution data of vehicles is needed. We use $f_{x,y}(t)$ to denote their distribution density function in slot t. The position of a UAV will affect the communication when the UAV transmits the output back to a vehicle, which in turn will affect the energy consumption of the task.

Without knowing the specifics of vehicles, UAV $i$ must guarantee that $P_i(t)= \left(2^{\frac{\lambda_{max}}{\gamma W}}-1\right) \frac{N_0 d_{x,y}^{i\,2}(t)}{g_0},(x,y)\in C_i$, where $\lambda_{max}$ denotes the maximum QoS requirement, and $d_{x,y}^{i\,2}(t)$ denote the Euclidean distance between UAV $i$ and vehicle located at $(x,y)$.

Inspired by \cite{7510870}, the average total transmit power of the UAVs in the network is given by
\begin{equation}
    \bar P(t)=\frac{\sum_{i\in\mathcal{B}}\int\!\!\!\int_{C_i}c_iP_i(t)f_{x,y}(t) \,dx\,dy}{\sum_{i\in\mathcal{B}}c_i}
\end{equation}

In \textit{lemma \ref{le1}} we proved that the communication energy consumption is increase with transmission power.
Our goal is to minimize the energy consumption in each slot, So $\bar{P}(t)$ is a suitable indicator to measure the efficiency of UAVs' position, and the following formulation is equivalent
\[
\min_{\textbf{L}(t)}\bar{P}(t)
\]
Minimizing $\bar{P}(t)$ is equivalent to minimize $Z_i(t)\triangleq\int\!\!\!\int_{C_i}c_iP_i(t)f_{x,y}(t) dx\,dy,\forall i\in \mathcal{B}$. 
By applying KKT conditions such that $\left\{\begin{aligned}
\frac{\partial Z_i(t)}{\partial x_i}=0\\
\frac{\partial Z_i(t)}{\partial y_i}=0
\end{aligned}\right.$, the optimal solution is given by
\begin{equation}\label{eq16}
    \left\{\begin{aligned}
        x_i^*(t)=\frac{\int\!\!\!\int_{C_i}xf_{x,y}(t)\,dx\,dy}{\int\!\!\!\int_{C_i}f_{x,y}(t)\,dx\,dy}\\
        y_i^*(t)=\frac{\int\!\!\!\int_{C_i}yf_{x,y}(t)\,dx\,dy}{\int\!\!\!\int_{C_i}f_{x,y}(t)\,dx\,dy}
    \end{aligned}\right.,i\in\mathcal{B}.
\end{equation}
Note that the UAV has limited hover speed,\ie (\ref{eq4}). Equation (\ref{eq16}) gives the ideal deployment distribution. For simplicity, let $L_i(t)$ denote $\left(x_i(t),y_i(t)\right)$. The practical deployment should be rewritten as
\begin{equation}\label{eq17}
    L_i^{*'}(t)=\left\{\begin{aligned}
    &L_i^{*}(t),\qquad\left|\left|L_i^{*}(t)-L_i^{*'}(t-1)\right|\right|\leq V_{max}\tau\\
    &L_i^{*'}(t-1)+\frac{V_{max}\tau \left(L_i^{*}(t)-L_i^{*'}(t-1) \right)}{\left|\left|L_i^{*}(t)-L_i^{*'}(t-1)\right|\right|},otherwise
    \end{aligned}\right.
\end{equation}
\section{NUMERICAL RESULTS} \label{sec5}

In this section, we evaluate the efficiency and performance of the proposed algorithm JOAoDR by presenting simulation results.
We consider a rectangular area which is the coverage of a BS.
Two UAVs are deployed and each of them has fixed service range. The coverage of them are partially overlapped.

In our simulation, we set the length of one time slot $\tau=5$ s.
The maximum speed of UAVs $V_{max}$ is $5$ m/s. The velocities of ground vehicles are randomly distributed in $[10,20]$ m/s. 
At the beginning of each time slot, tasks or requests are randomly generated. The input and output data sizes are set within $[4000,10000]$ Kb and $[2000,10000]$ Kb, respectively. The required QoS is in the interval $[256,768]$ Kb/s. For simplicity, transmission power of all vehicles are set to $10$ mW, and the CPU cycles needed to process one unit size of tasks are set to $1000$ cycle/bit. 
The system-specified parameters of UAVs are $\alpha=0.05$ and $\beta=0.9$. The channel gain $g_0$ is $-50$ dB and the noise power $N_0$ is $10^{-8}$ W. The ratio of actual transmission rate to channel capacity $\gamma=0.95$.
%\begin{comment}
\begin{figure}[tpb]
    \centering
    \subfigure[]{
    \label{fig2a}
    \begin{minipage}{\linewidth}
    \centering
    \includegraphics[scale=0.5]{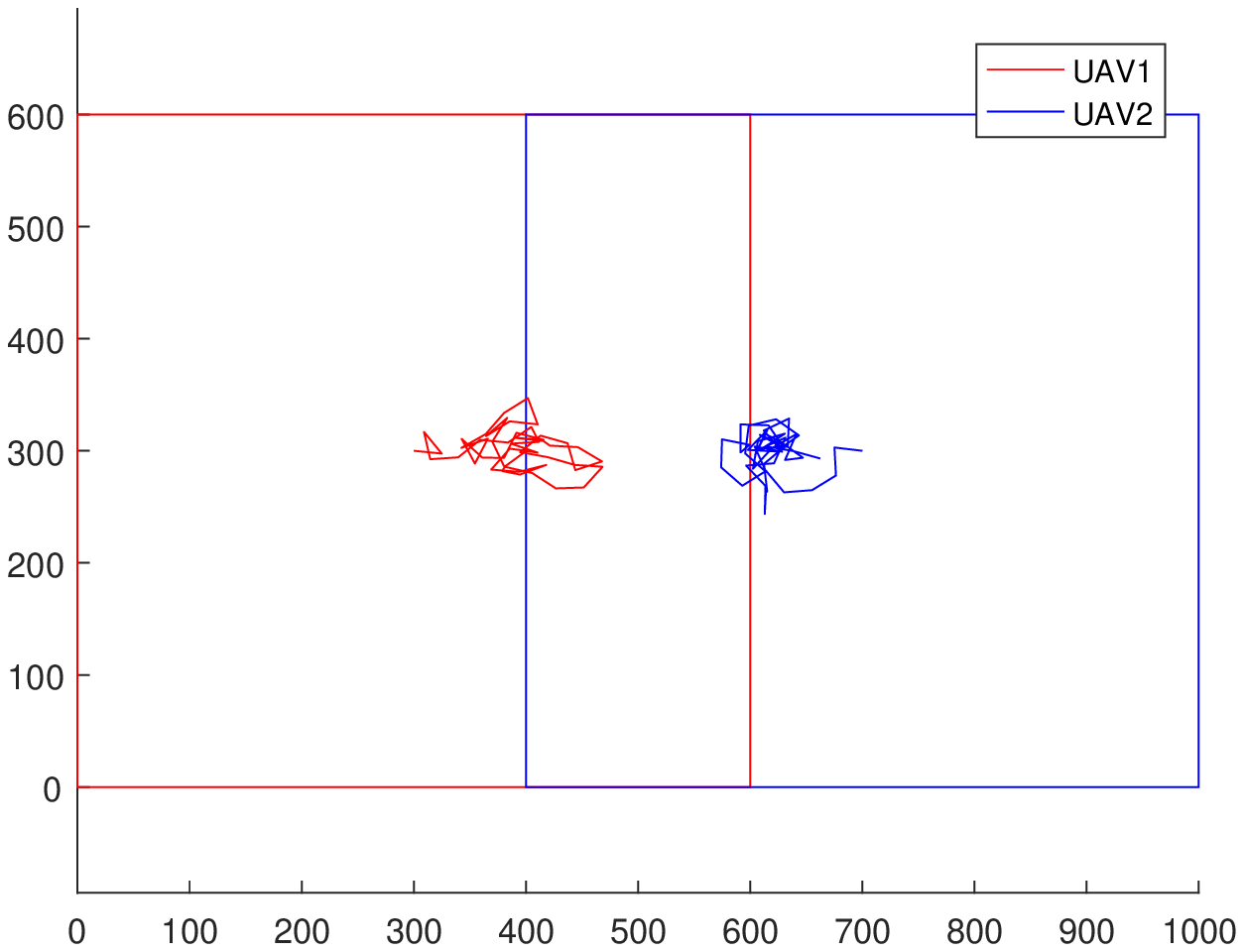}
    \end{minipage}
    }
    \subfigure[]{
    \label{fig2b}
    \begin{minipage}{\linewidth}
    \centering
    \includegraphics[scale=0.5]{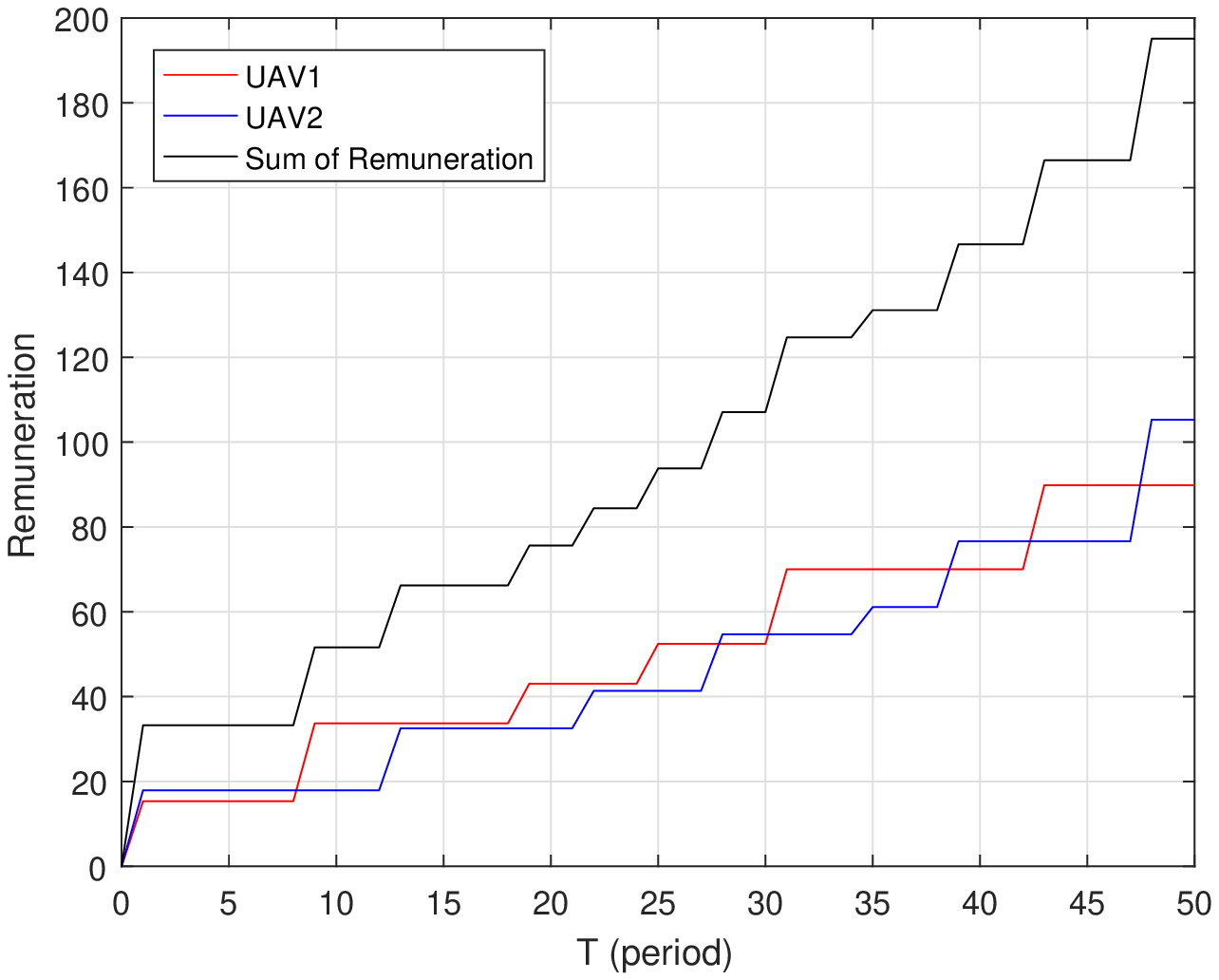}
    \end{minipage}
    }
    \subfigure[]{
    \label{fig2c}
    \begin{minipage}{\linewidth}
    \centering
    \includegraphics[scale=0.5]{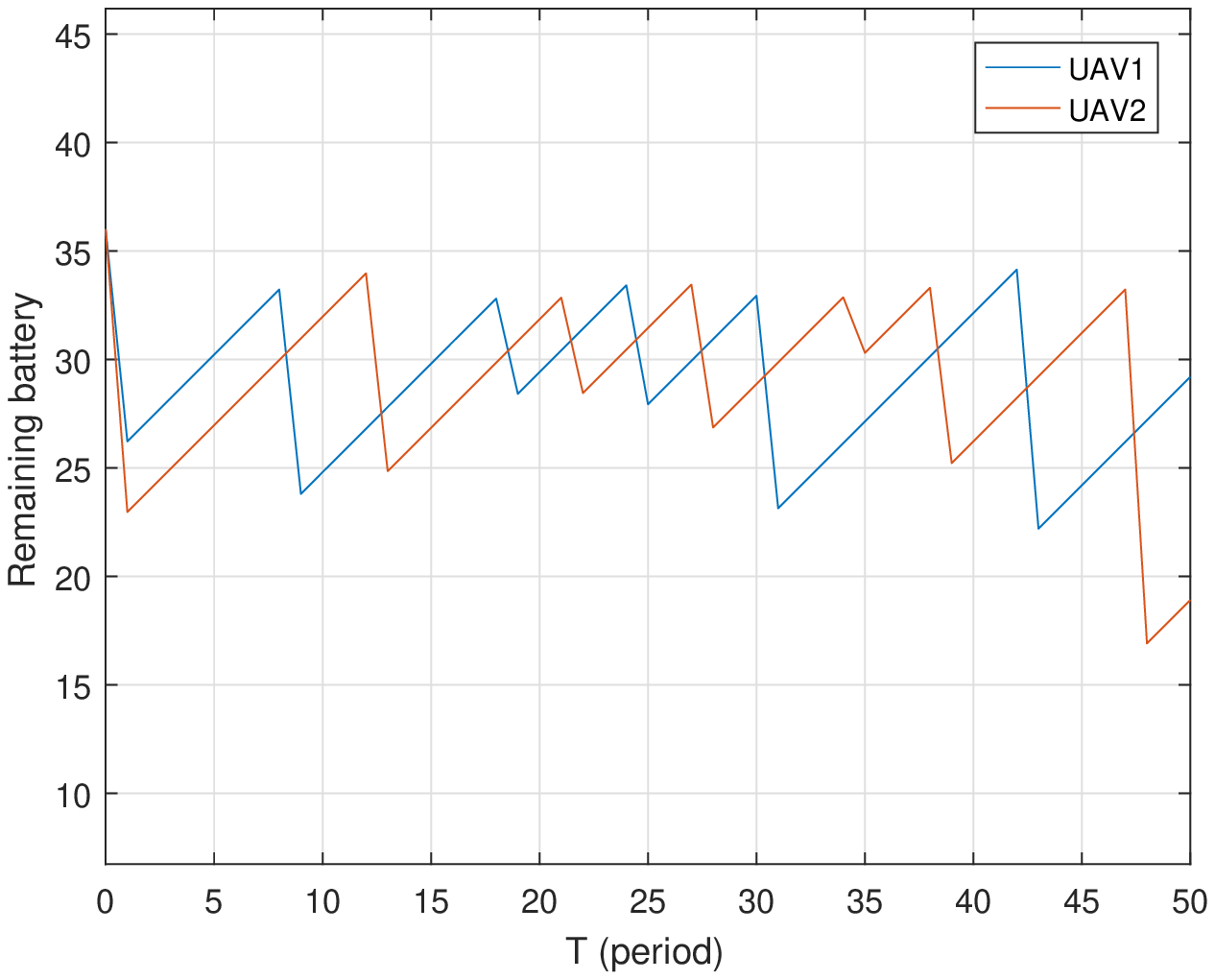}
    \end{minipage}
    }
    \subfigure[]{
    \label{fig2d}
    \begin{minipage}{\linewidth}
    \centering
    \includegraphics[scale=0.5]{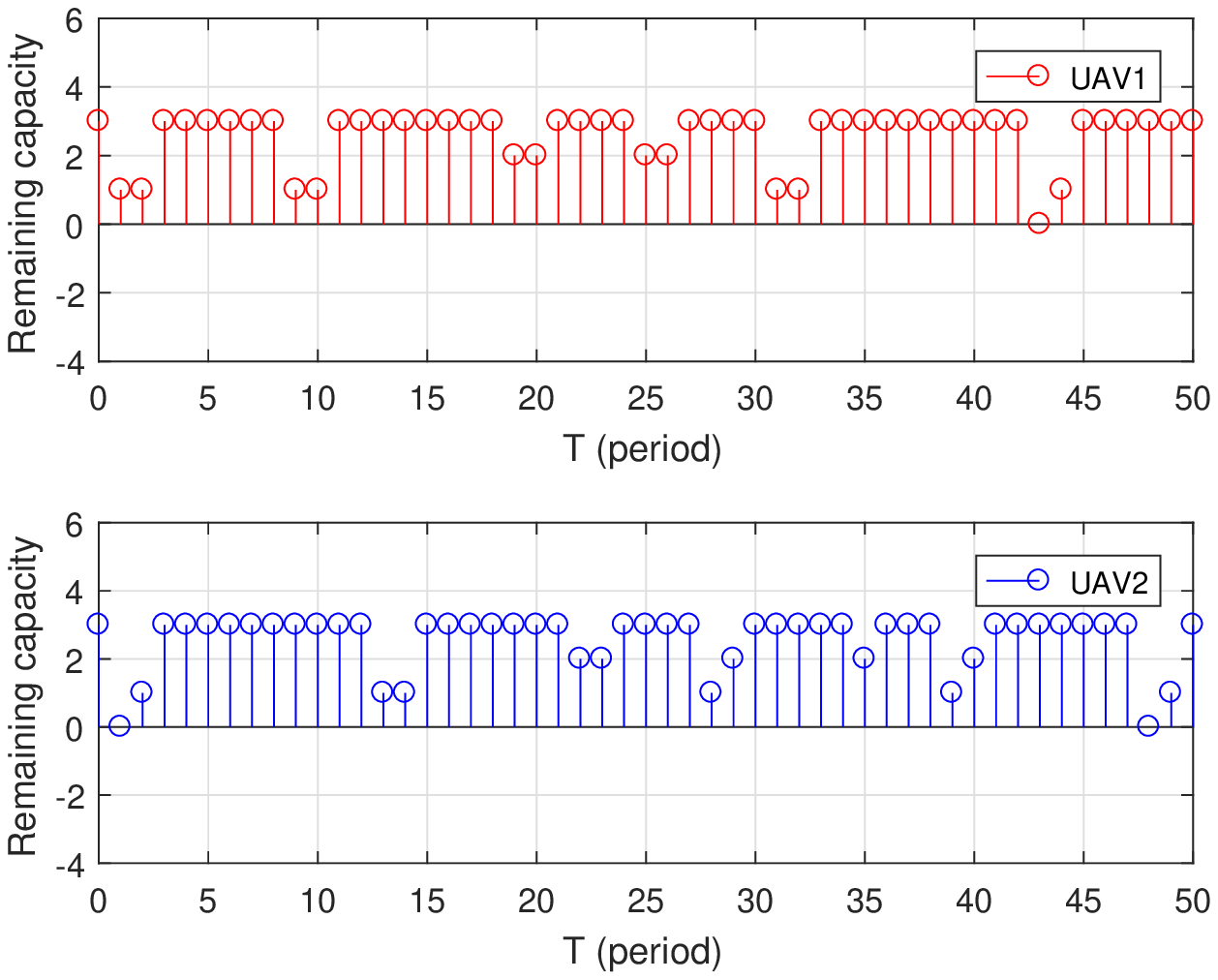}
    
    \end{minipage}
    }
    \caption{The trajectory, remuneration and consumption of UAVs}
    \label{fig2}
\end{figure}
%\end{comment}

Fig.\ref{fig2} depicts the trajectory and remuneration of both UAVs.
The red and blue rectangular areas in Fig. \ref{fig2a} denote the coverage of UAV $1$ and UAV $2$, respectively.
The advantage of the deployment in our JOAoDR is shown in Fig. \ref{fig4} which we will discuss later.
Fig. \ref{fig2b} illustrates two UAVs remuneration versus time periods $T$. In our setup, both UAVs hover at a fixed altitude of $300$ m, with energy harvest rate of $200$ mW and $V=2$. The curves present the form of the ladder to rise, because once a UAV processes a task and gets the reward, it will cost several time slots to finish it and collect energy for other tasks. In fact, the UAVs should be deployed at a relatively high altitude to avoid non-line-of-sight (NLoS) link.
Fig. \ref{fig2c} shows the remaining battery energy of each UAV versus $T$. The proposed JOAoDR can balance the energy consumption and rewards from the users to obtain a stable lone-term performance. it can be observed obviously that the remaining energy won't run out with a proper setup $\theta_i$ (In this case we set $\theta_i=36$ J) and verified the correctness of \textit{theorem \ref{th4}}.
Fig. \ref{fig2d} presents the remaining capacity of each UAV, \ie the number of tasks it can process simultaneously at the end of each slot. From the long-term perspective, UAVs will basically not fully put resources at a certain time slot, resulting in the lack of resources and energy in subsequent time slots to handle newly arrived high-value tasks.

\begin{figure}[tpb]
  \centering
     \includegraphics[scale=0.5]{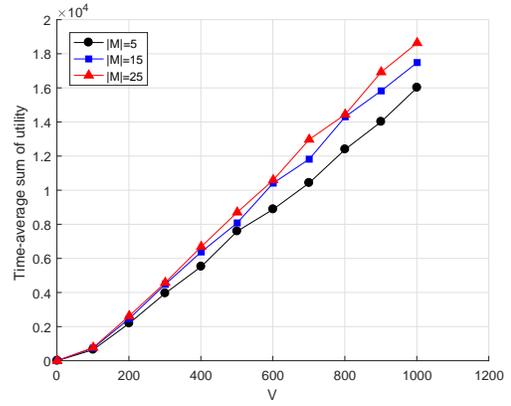}
    \caption{Time average utility}
    \label{fig3}
\end{figure}

Fig. \ref{fig3} presents time average utility of the operator system versus control parameter $V$ in three different cases where the number of users in the network is $5$, $15$ and $25$, respectively.
The system utility is consistent with the objective function of \textbf{P2}.
When the control parameter is small, the number of vehicles in the network cannot significantly affect the overall utility of operator system, because our algorithm is tend to concern the energy consumption of UAVs at that time. As the parameter increased, the slight reward gaps between different tasks will be magnified. At this time, the greater the number of users in the network, the greater the probability of tasks with higher rewards, so the system utility gaps will gradually increase.

\begin{figure}[tpb]
    \centering
    \includegraphics[scale=0.5]{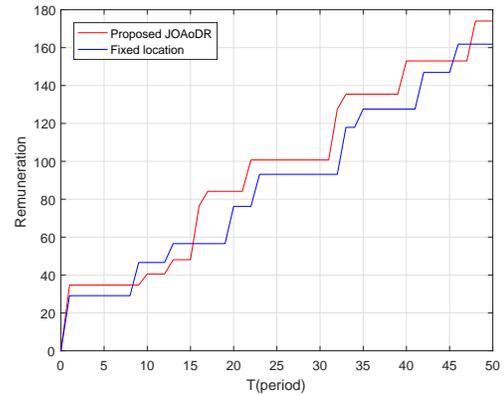}
    \caption{Remuneration comparison of different deployment methods}
    \label{fig4}
\end{figure}

Fig. \ref{fig4} illustrates the UAVs deployment method in our JOAoDR
outperforms the fixed deployment method as a benchmark. In the benchmark method, both UAVs are deployed at the geometric center of their coverage and do not move between slots.
Our algorithm can adjust the location of UAVs in each time slot according to historical distribution of ground vehicles. Because the distribution of vehicles has great similarity in time, this method can estimate the distribution of vehicles to a certain extent and let the UAV fly to the best location to reduce the energy consumption of communication with the ground vehicles. So that UAVs are able to serve more vehicles to increase remuneration.

\begin{figure}[tpb]
    \centering
    \includegraphics[scale=0.5]{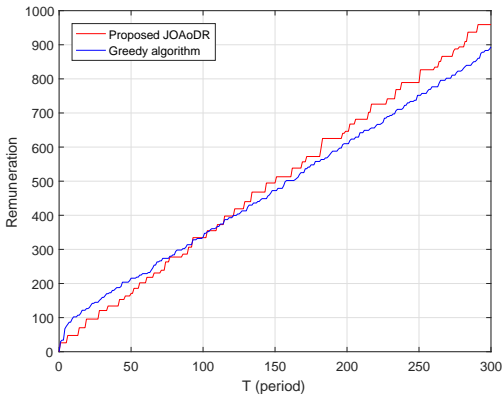}
    \caption{Remuneration comparison of different resource allocation methods}
    \label{fig5}
\end{figure}

In order for comparison, a greedy algorithm is introduced. It greedy seeks the tasks with the highest payments from vehicles, and process them if possible in every time slot.
It can be seen from Fig. \ref{fig5} that when the number of time periods is less than $90$, the greedy algorithm gains more than our algorithm, because the greedy algorithm utilize the energy and computing resources of the UAVs to process the most high-value tasks as much as possible. But over time, the superiority of our algorithm becomes more and more obvious, because our algorithm well balances the energy consumption of the UAVs and the completion rewards, and it can avoid the situation that one UAV cannot process high-value tasks at a certain moment due to insufficient computing resource. Compared with greedy algorithms, our JOAoDR greatly improves long-term performance.

\section{CONCLUSION} \label{sec6}
In this paper, we proposed a Lyapunov-based algorithm to balance the resource and rewards of the UAVs, and solved a long-term profit maximization problem in terms of the operator.
First, Lyapunov optimization was applied to transform origin problem.
Then our JOAoDR was proposed to optimize the deployment and the resource allocation of UAVs.
Numerical results demonstrated that our algorithm outperforms other benchmarks algorithm, and validated our solution.

\bibliographystyle{IEEEtran}
\bibliography{IEEEabrv,mylib}

\end{document}